\documentclass[10pt,letterpapert]{article}
\usepackage{graphicx}
\usepackage{amssymb}
\usepackage{epstopdf}
\usepackage{graphicx}
\usepackage{fancyhdr}
\usepackage{amsmath, amsthm, amssymb}
\usepackage{hyperref}
\usepackage{makeidx}
\usepackage{mathdesign}
\usepackage{color}
\definecolor{purple}{rgb}{.9,0,.9}

\def\XXint#1#2#3{{\setbox0=\hbox{$#1{#2#3}{\int}$}
        \vcenter{\hbox{$#2#3$}}\kern-.5\wd0}}

\usepackage{geometry}
\def\Nat{{\, \hbox{N \kern-1.25em I}\ \;}}
\def\Real{{\, \hbox{R \kern-1.25em I}\ \;}}
\def\Pos{{\, \hbox{P \kern-1.15em I}\ \;}}

\def\Int{{\, \hbox{\tenss Z \kern-1.1em Z} \,}}

\def\rfa{\qquad {\rm for \ all}\ \ }

\font\teneuf=eufm7 at 11pt
\font\seveneuf=eufm7
\font\fiveeuf=eufm5
\def\cE{{\cal E}} \def\cN{{\cal N}}\def\cS{{\cal S}}
\def\cT{{\cal T}}\def\cU{{\cal U}}\def\cV{{\cal V}}\def\cW{{\cal W}}

\def\bd{{\bf d}}
\def\be{{\bf e}}\def\bff{{\bf f}}

\def\bn{{\bf n}}

\def\bu{{\bf u}}\def\bv{{\bf v}}\def\bw{{\bf w}}
\def\bz{{\bf z}}
\def\bF{{\bf F}}

\def\bK{{\bf K}}\def\bL{{\bf L}}
\def\bR{{\bf R}}
\def\bS{{\bf S}}\def\bT{{\bf T}}

\def\b0{{\bf 0}}

\def\Lin{{\mathop{\rm Lin}\,}}

\newtheorem{theorem}{Theorem}[section]
\newtheorem{lemma}[theorem]{Lemma}

\def\ddd{d}

\def\beqn{\begin{equation}}
\def\eeqn{\end{equation}}

\def\Nat{\mathbb{N}}

\def\Real{\mathbb{R}}
\def\Pos{\mathbb{P}}

\numberwithin{equation}{section}

\def\Int{{\, \hbox{\tenss Z \kern-1.1em Z} \,}}
\def\Lin{{\rm Lin}}

\def\bLambda{{\boldsymbol \Lambda}}
\def\bGamma{{\boldsymbol \Gamma}}

\newcommand{\Fc}{\bF^\text{c}}

\newcommand{\da}{\,\text{d}a_y}
\newcommand{\dz}{\,\text{d}a_{\bz}}
\newcommand{\dr}{\,\text{d}r}

\newcommand{\half}{\frac{1}{2}}

\makeindex

\newcommand{\abs}[1]{|#1\rvert}

\title{Microphysical derivation of the\\ Canham--Helfrich free-energy density}
\author{Brian Seguin and Eliot Fried}
\begin{document}
\date{}

\maketitle

\begin{abstract}
\noindent
The Canham--Helfrich free-energy density for a lipid bilayer has drawn considerable attention. Aside from the mean and Gaussian curvatures, this free-energy density involves a spontaneous mean-curvature that encompasses information regarding the preferred, natural shape of the lipid bilayer. We use a straightforward microphysical argument to derive the Canham--Helfrich free-energy density. Our derivation (i) provides a justification for the common assertion that spontaneous curvature originates primarily from asymmetry between the leaflets comprising a bilayer and (ii) furnishes expressions for the splay and saddle-splay moduli in terms of derivatives of the underlying potential.
\\[8pt]
\end{abstract}

\section{Introduction}

Biomembranes are ubiquitous in nature as they form cell walls. An essential element of a biomembrane is a lipid bilayer, which is composed of lipid molecules. These molecules have a hydrophilic head and a hydrophobic tail. Due to these properties, when a large number of lipid molecules are placed in a solution they self assemble, under suitable conditions, into two-dimensional structures consisting of two leaflets (or monolayers). The lipid molecules are oriented so that the tails of the molecules in each leaflet are in contact with each other, while the heads are in contact with the ambient solution; see, for example, Lasic (1988). These two-dimensional structures often close to form vesicles and are usually between 50 nanometers and tens of micrometers in diameter but only a few nanometers thick, as observed by Luisi \& Walade (2000). Due to these dimensions, lipid bilayers are usually modeled as surfaces.

There is a long history of work dealing with specifying a free-energy density for a surface representing a lipid bilayer. Over forty years ago, Canham (1970), in an effort to explain the biconcave shape of red blood cells, proposed a bending-energy density dependent on the square of the mean curvature. A few years later, Helfrich (1973) also considered an elastic-energy density depending on the curvature of the surface associated with a lipid bilayer. In the Canham--Helfrich theory, the free-energy density $\psi$ of the lipid bilayer is of the form
\beqn\label{CHenergy}
\psi=\half \kappa(H-H_\circ)^2+\bar\kappa K,
\eeqn
where $H$ and $K$ are the mean and Gaussian curvatures, $\kappa$ and $\bar\kappa$ are the bending moduli, and $H_\circ$ is the spontaneous mean-curvature. The curvatures $H$ and $K$ are scalar invariants that describe the shape of the lipid bilayer. (See Appendix \ref{appsurf} for the definitions of $H$ and $K$.)  The splay modulus $\kappa$, which is always positive, characterizes the resistance to changes in mean curvature, while the saddle-splay modulus $\bar\kappa$ is related to the resistance to changes in Gaussian curvature. The spontaneous mean-curvature $H_\circ$ is determined by the preferred, natural, local shape of the lipid bilayer. Although Canham (1970) and Helfrich (1973) were the first to use a free-energy density of the form \eqref{CHenergy} for biomembranes, such expressions were considered by Poisson (1812), in the $H_\circ=0$ case, Germain (1821), and the Cosserat brothers (1909) in the study of elastic surfaces. Although not present in the Canham--Helfrich theory, it is also possible to introduce a spontaneous Gaussian-curvature $K_\circ$, as do Maleki et al.~(2012), in which case the last term on the right-hand side of \eqref{CHenergy} is replaced by $\bar\kappa(K-K_\circ)$. If $\bar\kappa$ and $K_\circ$ are both constant, the spontaneous Gaussian-curvature can be neglected since the free-energy density is only determined up to an additive constant.

As mentioned above, the spontaneous curvatures $H_\circ$ and $K_\circ$ are determined by the preferred, natural, local shape of the lipid bilayer. As Seifert (1997) notes, one of the main causes for spontaneous curvature is asymmetry between the two leaflets making up the bilayer. The asymmetry can manifest itself in different ways. For example, D\"{o}bereiner et al (1999) observe that spontaneous curvature may arise from differences in the chemical properties of the aqueous solution on the two sides of the lipid bilayer, while McMahon \& Gallop (1998) explain that lipid molecules with different configurations may also lead to nonzero $H_\circ$ and $K_\circ$.

Several derivations of the Canham--Helfrich free-energy density have appeared in the literature. For example, in a work concerned with lipid monolayers, Safran (1994) derived \eqref{CHenergy} by modeling the lipid molecules in a monolayer as springs. In Safran's (1994) derivation, the springs are allowed to change length with the bending of the bilayer but may not interact with one another. Granted that the lipid monolayer is incompressible, Safran (1994) expressed the length of a generic molecule in terms of the underlying mean and Gaussian curvatures. Upon making the additional assumption that the product of the spontaneous curvature with the natural length of the molecules is small, he found that the energy governing each spring takes a form consistent with \eqref{CHenergy}.

Paunov et al.~(2000) also used statistical methods to derive an energy for a lipid monolayer. Their expression for the energy consists of two terms: one associated with bending, which is taken from the work of Safran (1994), and the other associated with mixing. The mixing term accounts for interactions between the monolayer and the ambient solution. For the particular cases of spherical and cylindrical monolayers, Paunov et al.~(2000) found that, when expanded up to include only terms up to the square of the inverse radius of curvature, their expression for the energy of a monolayer agrees with \eqref{CHenergy}.

Following Ljunggren \& Eriksson (1995), who took the view that the bending energy of a bilayer is due to electrical forces, using the Debye--Huckel approximation, Winterhalter \& Helfrich (1988) obtained expressions for the bending energy of planar, cylindrical, and spherical lipid bilayers. As with the derivation of Paunov et al.~(2000), Winterhalter \& Helfrich (1988) considered only terms up to the square of the inverse radius of curvature and, in so doing, obtained expressions for the bending moduli and spontaneous curvature.

A derivation for the bending energy of a bilayer was also provided by Seifert (1997), who models the bilayer as a surface with molecules populating both sides, the molecular density being different on the sides.  Disregarding the structure of these molecules, Seifert (1997) assumes that the energy of the bilayer is determined by the extent to which these densities deviate from an equilibrium density. The densities are defined on the midsurface of each monolayer. To obtain an energy density defined on the surface, the two molecular densities are projected onto the surface. Since these projections depend on the curvature of the surface, the energy density of the surface inherits dependance on the curvature. By assuming that the curvature of the surface is small, Seifert (1997) finds an energy density for the surface of the Canham--Helfrich variety with no spontaneous curvature.

In the present paper, we follow an approach used by Keller \& Merchant (1991) to derive the energy of a liquid surface. To obtain their result, Keller \& Merchant (1991) expanded the energy of a mass of liquid in powers of the range $\ddd$ of the intermolecular potential. Beyond the first term of this expansion, which is volumetric, the second and third terms of this expansion are areal. Of the areal contributions to the energy, the first is due to surface tension and is proportional to the area of the underlying surface, and the second is proportional to the integral over the surface of $H^2-{\frac13}K$ and, thus, to the bending-energy density arising in the plate theory of F\"oppl~(1907) and von K\'arm\'an~(1910) with Poisson's ratio equal to $1/3$.

In our derivation of the free-energy density of a lipid bilayer, we neglect interactions between the lipid molecules and the ambient solution and instead focus on interactions between the lipid molecules themselves. These molecules are modeled as rigid rods, as in the work of Seguin \& Fried (2012). The rods, representing the molecules, are distributed over the surface that represents the lipid bilayer. By assuming that the lipid molecules are oriented perpendicular to the surface and interact with each other only when they are within a distance of $\ddd$, we obtain an expression for the free-energy density for the bilayer. Following the approach of Keller \& Merchant (1991), we then expand the free-energy density in powers of $\epsilon:=\ddd/\ell$, where $\ell$ is a characteristic length associated with the bilayer, and identify the Canham--Helfrich free-energy density as the coefficient of the term of order $\epsilon^2$.  This expansion rests on the assumption that the smallest possible curvature of the membrane is much larger than its thickness and, thus, our derivation does not hold for membranes that are undergoing crumpling, pitting, or budding.  Our derivation also neglects the effect of any molecular fluctuations.  While these fluctuations are important in crumpling, budding, and pitting events, as described by Lipowsky (1990) and Sackmann (1994), all of these events involve the membrane having large curvature and, thus, are not included in our theory.  Moreover, we view the Canham--Helfrich free-energy density as a mean-field description of a lipid bilayer. In classical mean-field theory, fluctuations are neglected.

This argument provides a direct and easy derivation of the Canham--Helfrich free-energy density. It also affords several observations. For one, if the lipid bilayer is symmetric, which occurs when its leaflets are of identical composition, and the lipid molecules are uniformly distributed, then the spontaneous curvatures $H_\circ$ and $K_\circ$ must vanish. Moreover, since our derivation gives concrete expressions for the bending moduli, those expressions might be used to gain insight regarding the sign of the saddle-splay modulus $\bar\kappa$.

Our paper is outlined as follows. In Section~\ref{sectmain}, we describe our approach to modeling lipid molecules, write down a free-energy density for a lipid bilayer, and state the result of expanding this energy in powers of $\ddd/\ell$.  The details of this expansion can be found in the Appendices. This is followed, in Section~\ref{sectint}, with a discussion of the resulting expansion and some consideration of special cases. In Section~\ref{sectfd}, we discuss possible extensions of the method used in Section~\ref{sectmain}, including some possible approaches to weakening the assumptions made there. Finally, in Section~\ref{sectsum}, we summarize our primary results and make some concluding remarks.

\section{Free-energy density of an isolated lipid bilayer}
\label{sectmain}

In this section, guided by the work of Keller \& Merchant (1991), we derive an expression for the free-energy density of a single closed (lipid) bilayer suspended in solution. Our result is based on the following assumptions:
\begin{itemize}
\item[(i)] The thickness of the bilayer is small relative to its average diameter.
\item[(ii)] The (lipid) molecules can be modeled as one-dimensional rigid rods.
\item[(iii)] The molecules do not tilt relative to the orientation of the bilayer.
\item[(iv)] Interactions between the bilayer and the solution are negligible.
\end{itemize}

On the basis of Assumption (i), which is commonly imposed in models for closed bilayers (Luisi \& Walade (2000)), we identify both leaflets of the bilayer with a single orientable surface $\cS$. This surface can adopt a large variety of shapes; however, being made up of molecules of a finite size, it cannot support arbitrarily large curvatures. Let $\ell$ denote the smallest stable radius of curvature that $\cS$ is capable of exhibiting. For each leaflet $i=1,2$ of the bilayer, we introduce a molecular number density $W_i$ defined on $\cS$ and measured per unit area of $\cS$. The total number of molecules in leaflet $i=1,2$ is then given by the integral
\beqn
\int_\cS W_i(y)\da,
\eeqn
where $\text{d}a_y$ denotes the area element on $\cS$. Taking $W_i$ to be defined on $\cS$ amounts to assuming that the centers of the lipid molecules of both leaflets lie on $\cS$, which is consistent with assuming that the bilayer is thin relative to its average diameter. In general, the number densities of the leaflets may differ.

On the basis of Assumption (ii), the configuration of each such molecule may be described by a point on $\cS$ and a unit-vector-valued director (Seguin \& Fried (2012)), with the point representing the center of the rod and the director representing the orientation of the rod. Without loss of generality, we suppose that the director tips point toward the headgroups of the molecules. We also assume that the interaction between a pair of molecules is governed by a potential that depends on the vector between the points and the directors and is restricted such that only molecules separated by distances less than some fixed cut-off distance $\ddd$ may interact. We assume that the cut-off distance $\ddd$ is small relative to the smallest radius of curvature $\ell$ the bilayer can support, so that $\ddd\ll\ell$ or, equivalently,
\beqn\label{small}
\epsilon:=\frac{\ddd}{\ell}\ll1.
\eeqn

Given molecules located at two points $x$ and $y$ on $\cS$ and corresponding directors $\be(x)$ and $\bff(y)$, we consider interaction potentials, with dimensions of energy, of the generic form
\beqn
\Phi^\ddd(x-y,\be(x),\bff(y))
=\phi(\epsilon^{-2}\abs{x-y}^2,(x-y)\cdot\be(x),(x-y)\cdot\bff(y),\be(x)\cdot\bff(y)),
\label{Psigeneric}
\eeqn
where, to ensure that rods separated by a distance greater than $\ddd$ may not interact, $\phi$ must satisfy
\beqn\label{ieradius}
\phi(s,a,b,c)=0
\quad\text{if $s\ge \ell$ for all $(a,b,c)\in\Real\times\Real\times[-1,1]$.}
\eeqn
Notice that $\Phi^d$ depends on $d$, while $\phi$ does not. That $\Phi^\ddd$ must depend on the vectors $x-y$, $\be(x)$, $\bff(y)$ only through the scalars $\abs{x-y}$, $(x-y)\cdot\be(x)$, $(x-y)\cdot\bff(x)$, and $\be(x)\cdot\bff(x)$ is shown by Seguin \& Fried (2012). This is also a consequence of the requirement that the potential be frame-indifferent.  Notice that the interaction energy $\phi(s,a,b,c)$ between two molecules can change if the head and tail of one of the molecules flips. Interaction potentials of this form may therefore account for differences between the polarities of the heads and tails of the molecules.  In taking $\Phi^\ddd$ to depend on the cutoff distance $\ddd$ as shown in \eqref{Psigeneric}, we follow Keller \& Merchant (1991). In addition to a potential $\Phi_{ii}^\ddd$ that accounts for interactions between molecules in each leaflet $i=1,2$, we introduce a potential $\Phi_{12}^\ddd=\Phi_{21}^\ddd$ that accounts for interactions between molecules belonging to different leaflets. Although the specific forms of these potentials may differ, they share the same general properties to the extent that they satisfy \eqref{Psigeneric} and \eqref{ieradius}.  A graph of a particular $\phi$ for fixed $a,b$, and $c$ is given in Figure \ref{fig0}. This function is based on the Lennard-Jones 6-12 potential. More tangible examples of interaction potentials are provided by Yuan (2010).

\begin{figure}[t]
\centering
\includegraphics[width=2.5in]{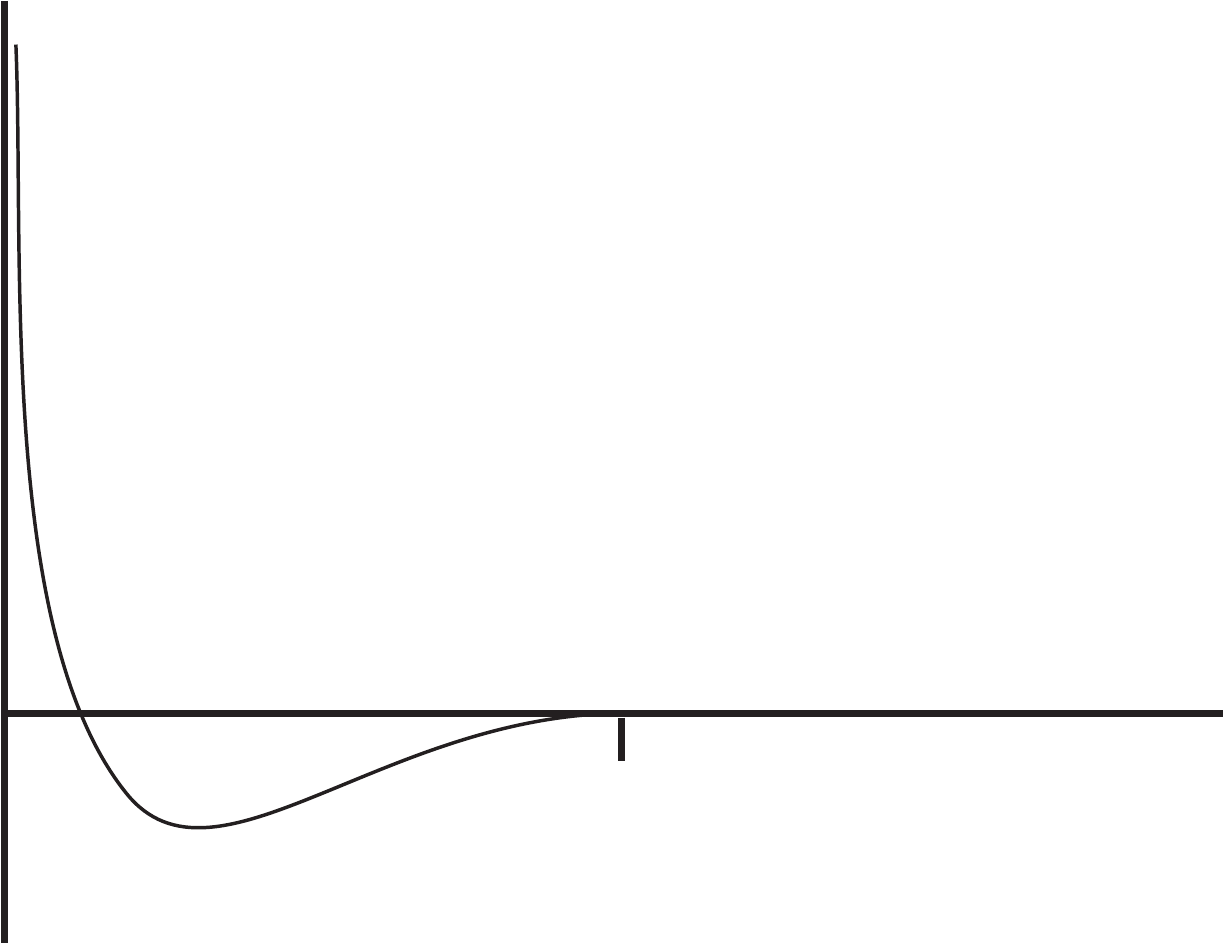}
\put(5,32){$s$}
\put(-91,20){$\ell$}
\put(-190,32){$0$}
\put(-190,130){$\phi$}
\caption{Graph of $s\mapsto\phi(s,a,b,c)$ for fixed $a,b$, and $c$.}
\label{fig0}
\end{figure}
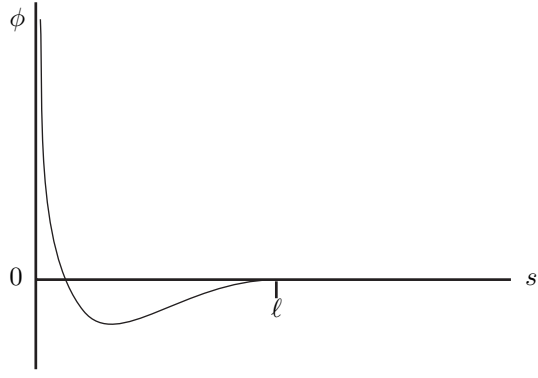

Without loss of generality, we choose to orient $\cS$ with the unit-normal field that points into the fluid adjacent to the head groups of leaflet 1 and denote that field by $\bn$. On the basis of Assumption (iii), it then follows that the directors of molecules in leaflets $1$ and $2$ coincide with $\bn$ and $-\bn$, respectively. Bearing in mind the cut-off property \eqref{ieradius}, define $\cS_\ddd(x)$ by
\beqn
\cS_\ddd(x):=\{y\in\cS:\abs{x-y}\leq\ddd\}.
\eeqn
The integrals
\beqn
\label{edtt}
\psi_{11}(x)=\epsilon^{-2}\int_{\cS_\ddd(x)}
\Phi_{11}^\ddd(x-y,\bn(x),\bn(y))W_{1}(x)W_{1}(y)\da
\eeqn
and
\beqn
\label{edbb}
\psi_{22}(x)=\epsilon^{-2}\int_{\cS_\ddd(x)}
\Phi_{22}^\ddd(x-y,-\bn(x),-\bn(y))W_{2}(x)W_{2}(y)\da,
\eeqn
then represent the contributions to the value of free-energy density of the bilayer at $x$ due, respectively, to interactions between the molecules of leaflets $1$ and $2$ at $x$ with all other molecules in their own leaflets.  Notice that the interaction energy $\Phi^d_{11}$ appearing in the integrand of \eqref{edtt} gives the interaction energy between a lipid molecule at $x$ and a lipid molecule at $y$. The two number densities are needed to account for the number of molecules at $x$ and the number of molecules at $y$. An analogous explanation explains the presence of the two densities in \eqref{edbb}. Similarly, the integrals
\beqn
\label{edtb}
\psi_{12}(x)=\epsilon^{-2}\int_{\cS_\ddd(x)}
\Phi_{12}^\ddd(x-y,\bn(x),-\bn(y))W_1(x)W_2(y)\da,
\eeqn
and
\beqn
\label{edbt}
\psi_{21}(x)=\epsilon^{-2}\int_{\cS_\ddd(x)}
\Phi_{21}^\ddd(x-y,-\bn(x),\bn(y))W_2(x)W_1(y)\da,
\eeqn
represent the contributions to the free-energy density of the bilayer at $x$ due, respectively, to interactions between the molecules of leaflet $1$ at $x$ with all the molecules of leaflet $2$ and to interactions between the molecules of leaflet $2$ at $x$ with all the molecules of leaflet $1$. Notice that interactions between the molecules in leaflet $1$ at $x$ with those in leaflet $2$ at $x$ are encompassed by $\psi_{12}(x)$; analogously, interactions between the molecules in leaflet $2$ at $x$ with those in leaflet $1$ at $x$ are encompassed by $\psi_{21}(x)$.

Finally, on the basis of Assumption (iv), the free-energy density of the bilayer at $x$ is determined completely by the sum of the four contributions defined in \eqref{edtt}--\eqref{edbt} and thus has the form
\beqn\label{energydecomp}
\psi(x)=\psi_{11}(x)+\psi_{22}(x)+\psi_{12}(x)+\psi_{21}(x).
\eeqn

On substituting \eqref{edtt}--\eqref{edbt} into the right-hand side of \eqref{energydecomp}, the goal is to expand $\psi(x)$ in powers of $\epsilon$ up to order $\epsilon^2$, following Keller \& Merchant (1991), with the objective of capturing the dependence on the curvature.  This expansion, the details of which appear in the Appendices, yields
\beqn\label{mainresult}
\psi=\psi_0+\half\kappa(H-H_\circ)^2+\bar\kappa(K-K_\circ),
\eeqn
where $\kappa, \bar\kappa, H_\circ$, and $K_\circ$ are defined in \eqref{kappa}--\eqref{Ko}, respectively. $\kappa$ and $\bar\kappa$ are of order $\epsilon^2$ and terms of order $\text{o}(\epsilon^2)$ are neglected. The quantities $\phi_0, \kappa, \bar\kappa, H_\circ$, and $K_\circ$ depend on the point $x$ in $\cS$. Thus, $\psi$ may depends on $x$ through not only the mean and Gaussian curvatures but also through the splay and saddle-splay moduli and the spontaneous mean and Gaussian curvatures.

\section{Interpretation of results}\label{sectint}

The result \eqref{mainresult} deserves some interpretation. The first term on the right-hand side of \eqref{mainresult} only depends on the interaction potentials and number densities at $x$. This term remains in the limit as $\ddd$ (or, equivalently, $\epsilon$) approaches zero. The remainder of \eqref{mainresult} is the Canham--Helfrich free-energy density, augmented by a spontaneous Gaussian-curvature term. Hence, \eqref{mainresult} provides a microphysical derivation of the Canham--Helfrich free-energy density \eqref{CHenergy} for a lipid bilayer. Moreover, the abbreviations \eqref{A}--\eqref{Ko} describe how the spontaneous curvatures $H_\circ$ and $K_\circ$ and the bending moduli $\kappa$ and $\bar\kappa$ depend on the interaction energies between the lipid molecules and the densities of the molecules in the two leaflets. The term $\psi_0$ is independent of the shape of the membrane and not part of the Canham--Helfrich free-energy density. However, due to implicit dependence of the number densities on temperature, concentration, and relevant electromagnetic fields, that term encompasses effects associated with ambient temperature, concentration, and relevant electromagnetic conditions.

Since the number densities $W_1$ and $W_2$ need not be constant, the spontaneous curvatures $H_\circ$ and $K_\circ$ and the bending moduli $\kappa$ and $\bar\kappa$ generally vary with position $x$ on the bilayer. Since the free-energy density can be determined only up to an arbitrary constant, if it happens that both $\bar\kappa$ and $K_\circ$ are constant, then the the term $-\bar{\kappa}K_\circ$ can be dropped without loss of generality. Otherwise, the presence of the spontaneous Gaussian-curvature $K_\circ$ cannot be neglected.

To obtain some additional insight regarding the spontaneous curvatures and the bending moduli, suppose that the two leaflets are identical, so that there is no leaflet asymmetry. This means that $\psi_{11}=\psi_{22}=\psi_{12}$ and $W_1=W_2$. From \eqref{C} and \eqref{Ho}, we see that in this case $C=H_\circ=0$ and, hence, that there is no spontaneous mean-curvature. Moreover, if the lipid molecules are evenly distributed on each leaflet so that the number densities $W_1$ and $W_2$ are constant, then $D=0$ and, hence, by \eqref{Ko}, the spontaneous Gaussian-curvature $K_\circ$ also vanishes. These results are in agreement with what other researchers have proposed regarding spontaneous curvature. Specifically, Seifert (1997) proposes that spontaneous curvature is a measure of the asymmetry of the two leaflets. McMahon \& Gallop (1998) suggest that spontaneous curvature is a result of the leaflets having different molecular compositions. In the framework presented here, possible differences in the compositions of the leaflets are encompassed by allowing $\psi_{11}$ and $\psi_{22}$ to differ.

Ideally, it should be possible to obtain an interaction energy for a given type of lipid molecule and use \eqref{A}--\eqref{Ko} to explicitly calculate the spontaneous curvatures and bending moduli for a bilayer consisting of these lipid molecules. For example, consider an interaction potential of the form
\beqn
\phi_{ij}(s,a,b,c):=\half p_{ij}(s)(c^2-1) \text{if $s\ge \ell$ for all $(s,a,b,c)\in\Real\times\Real\times\Real\times[-1,1]$},
\label{MSresult}
\eeqn
which is taken from the Ma\"ier--Saupe (1958) theory of liquid-crystals. Here, $p_{ij}$, $i,j=1,2$, are functions of the square of the normalized distance $s=\epsilon^{-2}|x-y|^2$ between the lipid molecules. In this case, it is easily confirmed that $\phi_{ij,2}=\phi_{ij,3}=0$ and, hence, by \eqref{A} and \eqref{Ho}, $H_\circ=0$. Thus, when applied to lipid molecules, the Ma\"ier--Saupe (1958) interaction potential leads to zero spontaneous mean-curvature. Without some concrete information about the functions $p_{ij}$ it is not possible to say anything about the characteristic features of the moduli $\kappa$ and $\bar\kappa$ and, thus, in particular, about the sign of $\bar\kappa$. 

Consider next the simple generalization
\begin{multline}
\phi_{ij}(s,a,b,c):=\half p_{ij}(s)(c^2-1+(a-a_{ij})^2+(b-b_{ij})^2)
\\[4pt]
\text{if $s\ge \ell$ for all $(s,a,b,c)\in\Real\times\Real\times\Real\times[-1,1]$}
\label{MSresult2}
\end{multline}
of \eqref{MSresult}, in which $a_{ij}$ and $b_{ij}$ are constants that influence the relative orientations of molecules. If $a_{ij}=b_{ij}=0$, then pairs of molecules  prefer to adopt orientations in which their directors are perpendicular to the line through their centers. If $a_{ij}$ and/or $b_{ij}$ is not zero, then pairs of molecules prefer to adopt orientations that are not perpendicular to the line through their centers. For interaction energies of the form \eqref{MSresult2}, \eqref{A} and \eqref{Ho} yield
\beqn
H_\circ:=\frac{\pi}{8(B+C)}\int_0^\ell\big[W_1^2(x)p_{11}(r^2)(a_{11}-b_{11})-W_2^2(x)p_{22}(r^2)(a_{22}-b_{22})\big]r^3\dr,\\
\eeqn
which shows that the spontaneous mean-curvature $H_\circ$ generally differs from zero when $a_{ii}\not=b_{ii}$ for $i=1$ or $i=2$. 

For another more complicated proposed interaction energy see the work of Yuan (2010).

\section{Possible further developments}\label{sectfd}

There are numerous ways in which the results of the previous two sections might be generalized. These are mainly related to the weakening of the assumptions made in Section \ref{sectmain}. Namely, we will consider possible ways to weaken Assumptions (ii) and (iii).

Let us begin with Assumption (iii), which concerns the orientation of the lipid molecules. Although Helfrich (1973) argues that any energetic contribution due to molecular tilt should be negligible and the lipid molecules making up a lipid bilayer are almost ubiquitously depicted as being perpendicular to the surface formed by the bilayer (see, for example, McMahon \& Gallop (2005)), there is no reason to rule out the possibility that molecules may tilt relative to the surface normal. Moreover, any such tilt need not be small. For simplicity, consider only leaflet $1$ and, for brevity, drop the subscript $1$. Analogous statements apply to leaflet $2$. A possible approach to modeling a lipid monolayer made up of tilted molecules is to assume that the orientations of those molecules on the surface to be given by a vector field $\bd$, so that $\bd(x)$ gives the orientation of the molecules at the point $x$ on the surface representing the monolayer. This vector field should satisfy
\beqn\label{conditions}
\abs{\bd}=1\qquad{\rm and}\qquad \bd\cdot\bn>1.
\eeqn
While the first condition in \eqref{conditions} constrains the lipid molecules to be of fixed length, the second condition ensures that the lipid molecules point roughly in the same direction as $\bn$. If \eqref{conditions}$_2$ did not hold, then either the lipid molecules lie flat on the surface ($\bd\cdot\bn=0$), which seems physically unrealistic, or the tail of the lipid molecules are in contact with the adjacent solution ($\bd\cdot\bn<0$), which, due to the chemical properties of the lipid molecules, is physically unattainable. Recall that while the tail of a lipid molecule is hydrophobic, the head is hydrophilic. The orientation $\bd$ should take the place of $\bn$ in expressions such as \eqref{edtt}. It is possible to find a second-order tensor field $\bK$ such that
\beqn\label{Kcond1}
\bd=\bK\bn.
\eeqn
In essence, the tensor $\bK$ describes the rotation of $\bn$ into $\bd$. Notice that the condition \eqref{Kcond1} is not sufficient to uniquely determine $\bK$, however, if the natural condition
\beqn
\bK(y)\bz=\bz\rfa \bz\in T_y\cS,\ y\in\cS
\eeqn
is also imposed, where $T_y\cS$ is the tangent space to $\cS$ at $y$, then $\bK$ is uniquely determined. With the substitution \eqref{Kcond1}, an argument analogous to that employed in Section \ref{sectmain} should go through. In particular, we conjecture that an expression for the free-energy density with the same form as \eqref{mainresult} should emerge but that the corresponding counterparts of the abbreviations \eqref{A}--\eqref{D} will involve $\bK$ and its surface gradient at $x$.

To also weaken Assumption (ii), namely the assumption that the lipid molecules are inextensible, two steps are necessary. First, the condition \eqref{conditions}$_1$ must be removed. Second, the potential governing the interaction between the molecules in \eqref{intpoten} must be replaced by
\beqn
\Phi^\ddd(x-y,\bd(x),\be(y))=
\phi(\epsilon^{-2}\abs{x-y}^2,(x-y)\cdot\bd(x),(x-y)\cdot\be(x),\bd(x)\cdot\be(x),\abs{\bd(x)}^2,\abs{\be(y)}^2).
\label{intpoten}
\eeqn
Once again, an analysis similar to that appearing in Section \ref{sectmain} should be feasible and we believe a result similar to \eqref{mainresult} might emerge from such an analysis.

It should also be feasible to relax Assumption (iv) by adding to the right-hand side of \eqref{energydecomp} a term that accounts for interactions between the lipid bilayer and the ambient solution. In pursuing such a generalization, it might make sense to follow the lead of Paunov et al.~(2000).

Steigmann (1999) developed a general theory of lipid bilayers, regarded as fluid surfaces with bending elasticity. Aside from accommodating a completely general dependence of the free-energy density on the mean and Gaussian curvatures, Steigmann's (1999) theory accounts for the areal stretch of the lipid bilayer relative to some reference placement. It would be interesting to see if it would be possible to use an argument like that contained in Section \ref{sectmain} to obtain a free-energy density that incorporates dependence on the areal stretch. To achieve this, it would perhaps be necessary to introduce a reference placement for the lipid molecules and then to account for interactions relative to that reference placement. Along with the other ideas mentioned in this section, we leave the details of any such analysis to a subsequent work.

\section{Summary}\label{sectsum}

In Section \ref{sectmain} we presented a microphysical derivation for the Canham--Helfrich form of the free-energy density of a lipid bilayer that includes not only spontaneous mean-curvature but also spontaneous Gaussian-curvature. Our approach was inspired by the work of Keller \& Merchant (1991). The resulting expressions for the spontaneous curvatures allowed us to remark in Section \ref{sectint} that if the two leaflets are identical and the lipid molecules are evenly distributed across the surface of the bilayer, then the natural shape of that surface is flat. Moreover, in Section \ref{sectfd}, we observed that it should be feasible to extend our method to allow for molecular tilt, molecular extension, interactions between the bilayer and the adjacent solution, and areal stretch.

\appendix
\section{Appendices}

In these appendices we present the necessary linear algebra and differential geometry necessary to carry out the expansions needed in Section \ref{sectmain}. We also include some details of the expansion.

\subsection{Linear algebra}

Here we recall some basic properties of tensors. Let $\cV$ be a finite-dimensional vector space equipped with an inner-product. For the purposes of this work, a tensor $\bT$ of order $m$ is a multilinear mapping that takes in $m$ vectors and gives a real number. A vector can be viewed as a first-order tensor and vice versa. A linear mapping from $\cV$ to itself can be viewed as a second-order tensor and vice versa; we denote the set of all such linear mappings by $\Lin(\cV,\cV)$. Other such statements regarding higher-order tensors can also be made. For example, a fourth-order tensor may be viewed as a bilinear mapping from $\cV$ to the space of second-order tensors and vice versa.

Given a tensor $\bT$ of order $m$ and a tensor $\bS$ of order $n$, it is possible to form their tensor product $\bT\otimes\bS$, which is a tensor of order $m+n$ defined by
\begin{multline}
\bT\otimes\bS(\bv_1,\dots,\bv_m,\bv_{m+1},\dots,\bv_{m+n})
\\[4pt]
=\bT(\bv_1,\dots,\bv_m)\bS(\bv_{m+1},\dots,\bv_{m+n})\rfa \bv_i\in\cV,\ i\in\{1,\dots,m+n\}.
\end{multline}
The tensor product is associative, meaning that $(\bT\otimes\bS)\otimes\bR=\bT\otimes(\bS\otimes\bR)$ and, thus, the use of parenthesis in this context may be dropped.

Just as it is possible to take the inner-product of two vectors, it is possible to take the inner-product of two tensors of the same order. The details of how this can be done in general will not be spelled out here. However, there are two useful identities that we will need regarding the inner-product of tensors. Namely, for any vector $\bz$ and second-order tensor $\bL$, the following hold:
\begin{align}
\label{TI1}(\bz\otimes\bz)\cdot\bL&=\bz\cdot\bL\bz\\
\label{TI2}(\bz\otimes\bz\otimes\bz\otimes\bz)\cdot(\bL\otimes\bL)&=(\bz\cdot\bL\bz)^2
\end{align}

Assume the dimension of $\cV$ is $2$. The following lemma, which can be established using polar coordinates and the integral gradient theorem, involving second- and fourth-order tensors will also be needed. Here, ${\bf 1}$ denotes the identity mapping on $\cV$ and, thus, is a second-order tensor.
\begin{lemma}
Let $g:\Pos\longrightarrow\Real$ be a continuous function. Then
\begin{align}
\label{TI3}\int_{\abs{\bz}\leq \ell}g(\abs{\bz})\bz\otimes\bz\,\emph{d}a_{\bz}
&=\bigg(\pi\int_0^\ell g(r)r^3\,\emph{d}r\bigg)\textrm{\bf 1},
\\[4pt]
\label{TI4}\int_{\abs{\bz}\leq \ell}g(\abs{\bz})\bz\otimes\bz\otimes\bz\otimes\bz\,\emph{d}a_{\bz}
&=\bigg(\frac{3\pi}{4}\int_0^\ell g(r)r^5\,\emph{d}r\bigg)\textbf{\bf 1}\otimes\textbf{\bf 1}.
\end{align}
\end{lemma}

\subsection{Geometry of surfaces} \label{appsurf}

Consider a surface $\cS$ in a three-dimensional Euclidean point space $\cE$ with associated vector space $\cV$. Given a point $y$ in $\cS$, let $\cT_y\cS$ denote the tangent space of $\cS$ at $y$. Let $\bn$ denote a mapping that determines a unit-normal to the surface at each point. Given a mapping $h:\cS\longrightarrow\cW$ defined on the surface that takes values in some vector space $\cW$, the surface gradient $\nabla^\cS h$ of $h$ can be defined by
\beqn\label{defsg}
\nabla^\cS_x h:=\nabla_x h^e (\textbf{1}-\bn(x)\otimes\bn(x))\rfa x\in\cS,
\eeqn
where $h^e$ is an extension of $h$ to a neighborhood of $x$ and $\nabla_x h^e$ is the classical three-dimensional gradient of this extension at $x$. It can be shown that the definition of the surface gradient is independent of the extension used on the right-hand side of \eqref{defsg}.

Of particular interest is the opposite $\bL:=-\nabla^\cS\bn$ of the surface gradient $\nabla^\cS\bn$ of $\bn$---called the curvature tensor, which is a second-order tensor field defined on $\cS$. The curvature tensor is symmetric and has two scalar invariants: the mean curvature $H$ and Gaussian curvature $K$ defined by
\begin{align}
\label{H}H&:=\half {\rm tr}\,\bL,\\
\label{K}K&:=\half[({\rm tr}\,\bL)^2-{\rm tr}(\bL^2)].
\end{align}
If $\kappa_1$ and $\kappa_2$ are the two nontrivial eigenvalues of $\bL$, often called the principle curvatures, then
\begin{align}
H&=\half(\kappa_1+\kappa_2),\\
K&=\kappa_1\kappa_2.
\end{align}

A useful identity involving the curvature tensor is provided by the following result.
\begin{lemma}
For $\bz$ in $\cV$ and $x$ in $\cS$,
\beqn\label{needed}
\bn(x)\cdot(\nabla_x^\cS\bL\bz)\bz=\abs{\bL(x)\bz}^2.
\eeqn
\end{lemma}

\begin{proof}
Using the definition and symmetry of $\bL$, we have
\beqn\label{intermediate}
\bn(x)\cdot\bL(x)\bz=0\rfa x\in\cS.
\eeqn
Taking the surface gradient of both sides of \eqref{intermediate} in the direction $\bz$, using the product rule, and evaluating the result at $x$ yields
\beqn
-\bL(x)\bz\cdot\bL(x)\bz+\bn(x)\cdot(\nabla_x^\cS\bL\bz)\bz=0,
\eeqn
which is \eqref{needed}.
\end{proof}

From here on fix a point $x$ in $\cS$. It is possible to parameterize a small neighborhood of $\cS$ at $x$ using a local curving $f$. Put $\cT:=T_x\cS$ and $\cN:=(T_x\cS)^\perp$, and let $\cU$ be an open neighborhood of $\textbf{0}$ in $\cV$. A local curving $f:\cU\longrightarrow\cE$ of $\cS$ at $x$ has the following properties:
\begin{enumerate}
\item $f(\textbf{0})=x$,
\item $f(\bu)\in\cS$ for all $\bu\in\cT\cap\cU$,
\item\label{item3} $f(\bu)-(x+\bu)\in\cN$ for all $\bu\in\cT\cap\cU$,
\item\label{item4} $f(\bu+\bw)=f(\bu)+\bw$ for all $\bu\in\cT\cap\cU$ and $\bw\in\cN\cap\cU$,
\item $\nabla_{\textbf{0}} f=\textbf{1}$.
\end{enumerate}

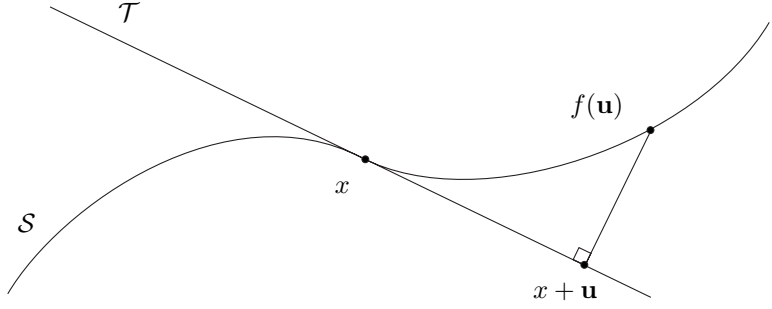
\begin{figure}[t]
\centering
\includegraphics[width=4in]{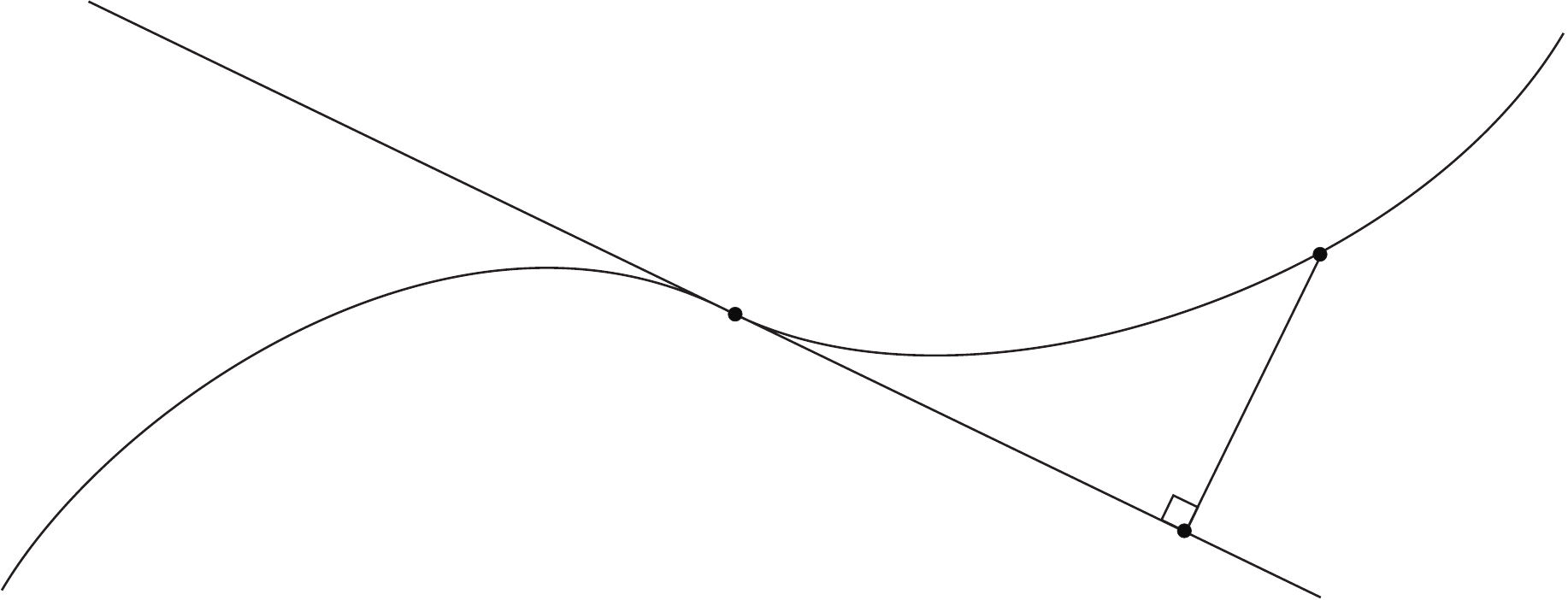}
\put(-165,40){$x$}
\put(-285,25){$\cS$}
\put(-247,105){$\cT$}
\put(-90,0){$x+\bu$}
\put(-76,70){$f(\bu)$}
\caption{Depiction of a local curving for a surface $\cS$ at $x$.}
\label{fig1}
\end{figure}

The gradients of $f$ at $\textbf{0}$ describe the local shape of the surface at $x$. Consider the mappings
\begin{align}
\bF&:\cU\longrightarrow\Lin(\cV,\cV),\\
\bLambda&:\cV\times\cV\longrightarrow\cV,\\
\bGamma&:\cV\times\cV\longrightarrow\Lin(\cV,\cV),
\end{align}
defined by 
\begin{align}
\label{F}\bF&:=\nabla f,\\
\bLambda(\bu,\bv)&:=[(\nabla_{\textbf{0}}\bF)\bu]\bv\rfa \bu,\bv\in\cV,\\
\bGamma(\bu,\bv)\bw&:=([(\nabla_{\textbf{0}}\nabla\bF)\bu]\bv)\bw\rfa \bu,\bv,\bw\in\cV.
\end{align}
The mapping $\bF$ is a second-order tensor field while the quantities $\bLambda$ and $\bGamma$ can be viewed as third- and fourth-order tensors, respectively. Using Item \ref{item3}, it can be shown that
\beqn\label{ranges}
\bLambda(\bz_1,\bz_2)\in\cN\qquad{\rm and}\qquad \bGamma(\bz_1,\bz_2)\bz_3\in\cN\rfa \bz_1,\bz_2,\bz_3\in\cT;
\eeqn
further, using Item \ref{item4}, it can be shown that
\beqn\label{nullsp}
\bLambda(\bz_1,\bn)=\textbf{0}\qquad{\rm and}\qquad \bGamma(\bz_1,\bz_2)\bn=\textbf{0}\rfa \bz_1,\bz_2\in\cT.
\eeqn
It follows from \eqref{ranges}--\eqref{nullsp} that
\beqn\label{traces}
{\rm tr}\, \bLambda\bz_1=0\qquad{\rm and}\qquad {\rm tr}(\bGamma(\bz_1,\bz_2))=0\rfa \bz_1,\bz_2\in\cT.
\eeqn
The curvature tensor $\bL(x)$ is related to $\bLambda$ by
\beqn\label{LambdaL}
\bLambda(\bz_1,\bz_2)=\bn(\bz_1\cdot\bL(x)\bz_2)\rfa \bz_1,\bz_2\in\cT.
\eeqn
It follows from \eqref{nullsp} and \eqref{LambdaL} that
\beqn
\bLambda\bz=\bn\otimes(\bL(x)\bz)\rfa \bz\in\cT.
\eeqn

\subsection{Useful approximations}

In this subsection we delineate the key expansions that will be needed in the asymptotic expansion of the energy. The three basic expansions, which follow from the results stated in Appendix \ref{appsurf}, are
\begin{align}
\label{apf}f(\epsilon\bz)&=x+\epsilon\bz+\frac{\epsilon^2}{2}(\bz\cdot\bL(x)\bz)\bn(x)+\text{o}(\epsilon^2),\\
\label{apn}\bn(f(\epsilon\bz))&=\bn-\epsilon\bL(x)\bz-\frac{\epsilon^2}{2}[(\nabla_x^\cS\bL)\bz]\bz+\text{o}(\epsilon^2),\\
\label{apF}\bF(\epsilon\bz)&=\textbf{1}+\epsilon\bn(x)\otimes\bL(x)\bz+\frac{\epsilon^2}{2}\bGamma(\bz,\bz)+\text{o}(\epsilon^2).
\end{align}
Using \eqref{needed} and \eqref{apf}--\eqref{apn}, it transpires that
\begin{align}
\label{ap1}\epsilon^{-2}\abs{x-f(\epsilon\bz)}^2&=\abs{\bz}^2+\frac{\epsilon^2}{4}(\bz\cdot\bL(x)\bz)^2+\text{o}(\epsilon^2),\\
\label{ap2}(x-f(\epsilon\bz))\cdot\bn(x)&=-\frac{\epsilon^2}{2}(\bz\cdot\bL(x)\bz)+\text{o}(\epsilon^2),\\
\label{ap3}(x-f(\epsilon\bz))\cdot\bn(f(\epsilon\bz))&=\frac{\epsilon^2}{2}(\bz\cdot\bL(x)\bz)+\text{o}(\epsilon^2),\\
\label{ap4}\bn(f(\epsilon\bz))\cdot\bn(x)&=1-\frac{\epsilon^2}{2}\abs{\bL(x)\bz}^2+\text{o}(\epsilon^2).
\end{align}
It can be shown that the cofactor $\Fc$ of $\bF$ is given by (Gurtin et al.~(2010))
\beqn\label{cofF}
\Fc:=\big[\bF^2-({\rm tr}\,\bF)\bF+\half[({\rm tr}\,\bF)^2-{\rm tr}(\bF^2)]\textbf{1}\big]^{\top}.
\eeqn
Using \eqref{needed}, \eqref{traces}, \eqref{apn}--\eqref{apF}, and \eqref{cofF}, it follows that
\begin{align}
\label{aparea}\bn(f(\epsilon\bz))\cdot\Fc(\epsilon\bz)\bn(x)=1+\frac{\epsilon^2}{2}\abs{\bL(x)\bz}^2+\text{o}(\epsilon^2).
\end{align}
Given a scalar-valued function $W$ defined only on the surface, if we write
\beqn
W'(x):=\nabla_x^\cS W\qquad{\rm and}\qquad W''(x):=\nabla_x^\cS\nabla^\cS W,
\eeqn
then
\begin{align}
\label{apW}W(f(\epsilon\bz))=W(x)+\epsilon W'(x)\bz+\frac{\epsilon^2}{2}\bz\cdot W''(x)\bz+\text{o}(\epsilon^2).
\end{align}

\subsection{Expansion of the energy}

Here we use the results of the previous two subsections to perform the expansion of \eqref{energydecomp} in powers of $\epsilon$. Since the domains of integration of the integrals \eqref{edtt}--\eqref{edbt} depend on $\epsilon$ through $d$, this cannot presently be achieved. To proceed, we introduce a change of variables that transfers this dependence to the integrand. To achieve this, we parameterize $\cS_\ddd(x)$ using the notion of a local curving $f$ introduced in Appendix \ref{appsurf}, which is made possible by \eqref{small}. In particular, it is convenient to introduce the set
\beqn
\cT_\ell(0):=\{\bz\in T_x\cS:\abs{\bz}\leq\ell\},
\eeqn
where $T_x\cS$ denotes the tangent space of $\cS$ at $x$. Focusing on \eqref{edtt} and using using the abbreviation \eqref{F}, the change of variables $y=f(\epsilon\bz)$ yields
\beqn\label{edttcv}
\psi_{11}(x)=\epsilon^{-2}\int_{\cT_\ell(0)}\Phi_{11}^\ddd (x-f(\epsilon\bz),\bn(x),\bn(f(\epsilon\bz)))W_{1}(x)W_{1}(f(\epsilon\bz))\cdot\bn(f(\epsilon\bz))\cdot\Fc(\epsilon\bz)\bn(x)\dz.
\eeqn
Let us focus our attention on the quantity
\begin{multline}
\Phi_{11}^\ddd(x-f(\epsilon\bz),\bn(x),\bn(f(\epsilon\bz))=
\\[4pt]
\phi_{11}(\epsilon^{-2}\abs{x-f(\epsilon\bz)}^2,(x-f(\epsilon\bz))\cdot\bn(x),
(x-f(\epsilon\bz))\cdot\bn(f(\epsilon\bz)),\bn(x)\cdot\bn(f(\epsilon\bz))).
\end{multline}
For all $0\le s\le\ell$, put 
\beqn
\label{appsi0}
\phi_{11,0}(s):=\phi_{11}(s^2,0,0,1),
\eeqn
and
\beqn
\label{appsi}
\phi_{11,k}(s):=\frac{\phi_{11}(\xi_1,\xi_2,\xi_3,\xi_4)}{\partial\xi_k}
\bigg|_{(\xi_1,\xi_2,\xi_3,\xi_4)\mskip1.5mu=\mskip1.5mu(s^2,0,0,1)},
\qquad
k\in\{1,2,3,4\}.
\eeqn
Using the expansions \eqref{ap1}--\eqref{ap4} and the notation \eqref{appsi0}--\eqref{appsi}, we have
\begin{multline}
\Phi_{11}^\ddd(x-f(\epsilon\bz),\bn(x),\bn(f(\epsilon\bz))
= \psi_{11,0}(\abs{\bz})
\\[4pt]
+\frac{\epsilon^2}{2}\Big(\half\phi_{11,1}(\abs{\bz})(\bz\cdot\bL(x)\bz)^2-\phi_{11,2}(\abs{\bz})(\bz\cdot\bL(x)\bz)\\[4pt]+\phi_{11,3}(\abs{\bz})(\bz\cdot\bL(x)\bz)+\phi_{11,4}(\abs{\bz})\abs{\bL(x)\bz}^2\Big)+\text{o}(\epsilon^2).
\end{multline}
where $\bL=-\nabla^{\cS}\bn$ denotes the curvature tensor of $\cS$ (see Appendix~\ref{appsurf}). Combining this expansion with \eqref{aparea} and using the expansion \eqref{apW} for $W_1$ yields an expansion for the right-hand side of \eqref{edttcv}:
\begin{multline}
\phi_{11}(x)=\int_{T_x\cS_1(0)}\Big(\phi_{11,0}(\abs{\bz})W^2_1(x)
+\frac{\epsilon^2}{2}W_1(x)\big[W_1(x)(\phi_{11,0}(\abs{\bz})\\[4pt]-\phi_{11,4}(\abs{\bz}))\abs{\bL(x)\bz}^2+\phi_{11,0}(\abs{\bz})\bz\cdot  W''_1(x)\bz+\frac{1}{2}W_1(x)\psi_{11,1}(\abs{\bz})(\bz\cdot\bL(x)\bz)^2\\[4pt]
\label{psi2}+W_1(x)(\phi_{11,3}(\abs{\bz})-\psi_{11,2}(\abs{\bz}))(\bz\cdot\bL(x)\bz)\big]\Big)\dz+\text{o}(\epsilon^2).
\end{multline}
Similarly, the right-hand sides of \eqref{edbb}--\eqref{edbt} can be expanded in powers of $\epsilon$, using notation similar to that of \eqref{appsi0}--\eqref{appsi}. Putting these expansions together, we obtain
\begin{align}
\nonumber\psi(x)&=\int_{T_\ell(0)}\big[\phi_{11,0}(\abs{\bz})W^2_1(x)+\phi_{22,0}(\abs{\bz})W^2_2(x)+2\phi_{12,0}(\abs{\bz})W_1(x)W_2(x)\big]\dz\\
\nonumber&\qquad+\frac{\epsilon^2}{2}\int_{T_\ell(0)}\big[W_1^2(x)(\phi_{11,3}(r)-\phi_{11,2}(r))-W_2^2(x)(\phi_{22,3}(r)\\
\nonumber&\qquad\qquad-\phi_{22,2}(r))\big] (\bz\cdot\bL(x)\bz) \dz\\
\nonumber&\qquad+\frac{\epsilon^2}{2}\int_{T_\ell(0)}\big[W_1^2(x)(\phi_{11,0}(\abs{\bz})-\phi_{11,4}(\abs{\bz}))+W^2_2(x)(\phi_{22,0}(\abs{\bz})\\
\nonumber&\qquad\qquad-\phi_{22,4}(\abs{\bz}))+2W_1(x)W_2(x)(\phi_{12,0}(\abs{\bz})+\phi_{12,4}(\abs{\bz}))\big]\abs{\bL(x)\bz}^2\dz\\
\nonumber&\qquad+\frac{\epsilon^2}{2}\int_{T_\ell(0)}\big[W_1^2(x)\phi_{11,1}(\abs{\bz})+W^2_2(x)\phi_{22,1}(\abs{\bz})\\
\nonumber&\qquad\qquad+2W_1(x)W_2(x)\phi_{12,1}(\abs{\bz})\big](\bz\cdot\bL(x)\bz)^2\dz\\
\nonumber&\qquad+\frac{\epsilon^2}{2}\int_{T_\ell(0)}\big[(W_1(x)\phi_{11,0}(\abs{\bz})+W_2(x)\phi_{12,0}(\abs{\bz})) (\bz\cdot W^{\prime\prime}_1(x)\bz)\\
\label{collectingterms}&\qquad\qquad+(W_2(x)\psi_{22,0}+W_1(x)\phi_{12,0}(\abs{\bz}))(\bz\cdot W^{\prime\prime}_2(x)\bz)]\dz.
\end{align}
Notice that the first term on the right-hand side of \eqref{collectingterms} is independent of $\epsilon$, while the others are proportional to $\epsilon^2$. Of the terms involving $\epsilon$, the first is linear is $\bL$ and the second and third are quadratic in $\bL$. The last term in \eqref{collectingterms} is independent of $\bL$. Using the identities \eqref{TI1}--\eqref{TI4}, the abbreviations
\begin{align}
\label{phi0}\psi_0(x)&:=2\pi\int_0^\ell\big[\phi_{11,0}(r)W^2_1(x)+\phi_{22,0}(r)W^2_2(x)+2\phi_{12,0}(r)W_1(x)W_2(x)\big]r\dr,\\
\label{A}A(x)&:=\frac{\pi}{2}\int_0^\ell\big[W_1^2(x)(\phi_{11,3}(r)-\phi_{11,2}(r))-W_2^2(x)(\phi_{22,3}(r)-\phi_{22,2}(r))\big]r^3\dr,\\
\nonumber B(x)&:=\frac{\pi}{2}\int_0^\ell\big[W_1^2(x)(\phi_{11,0}(r)-\phi_{11,4}(r))+W^2_2(x)(\phi_{22,0}(r)-\phi_{22,4}(r))\\
\label{B}&\qquad+2W_1(x)W_2(x)(\phi_{12,0}(r)+\phi_{12,4}(r))\big]r^3\dr,\\
\label{C}C(x)&:=\frac{3\pi}{16}\int_0^\ell\big[W_1^2(x)\phi_{11,1}(r)+W^2_2(x)\phi_{22,1}(r)+2W_1(x)W_2(x)\phi_{12,1}(r)\big]r^5\, dr,\\
\nonumber D(x)&:=\frac{\pi}{2}\int_0^\ell\big[(W_1(x)\phi_{11,0}(r)+W_2(x)\phi_{12,0}(r)){\rm tr} (W^{\prime\prime}_1(x))\\
\label{D}&\qquad+(W_2(x)\phi_{22,0}+W_1(x)\phi_{12,0}(r)){\rm tr} (W^{\prime\prime}_2(x))\big]r^3\dr,
\end{align}
and motivated by the grouping in \eqref{collectingterms}, we find (on suppressing explicit dependence on $x$) that
\beqn
\psi=\psi_0+\big[ A\,\text{tr}\, \bL+B\,\text{tr}(\bL^2)+C(\text{tr}\, \bL)^2+D\big]\epsilon^2+\text{o}(\epsilon^2).
\eeqn
Further, upon using \eqref{H}--\eqref{K} and the additional abbreviations
\begin{align}
\label{kappa}\kappa&:=8\epsilon^2(B+C),\\
\label{kappabar}\bar\kappa&:=-2\epsilon^2B,\\
\label{Ho}H_\circ&:=\frac{A}{4(B+C)},\\
\label{Ko}K_\circ&:=\frac{4D(B+C)-A^2}{8B(B+C)},
\end{align}
allows us to write the expansion of the right-hand side of \eqref{energydecomp} as\footnote{See \eqref{H} and \eqref{K} for the definitions of the mean $H$ and Gaussian $K$ curvatures.}
\beqn\label{mainresulta}
\psi=\psi_0+\half\kappa(H-H_\circ)^2+\bar\kappa(K-K_\circ),
\eeqn
where the correction of $\text{o}(\epsilon^2)$ has been neglected.

\subsection*{Acknowledgement}  

We thank Mohsen Maleki for very fruitful discussions.



\end{document}